\documentclass[reqno]{amsart}
\usepackage{amsmath}
\usepackage{amsfonts}
\usepackage{amssymb}
\usepackage{amsthm}
\usepackage{newlfont}
\usepackage{graphicx}
\usepackage{amscd}

\theoremstyle{plain}
\newtheorem{Th}{Theorem}[section]
\newtheorem{Cor}[Th]{Corollary}
\newtheorem{Lem}[Th]{Lemma}
\newtheorem{Prop}[Th]{Proposition}
\theoremstyle{definition}
\newtheorem{Def}{Definition}[section]

\theoremstyle{remark}
\newtheorem*{Rem}{Remark}
\numberwithin{equation}{section}

\newcommand{\DD}{{\mathbb D}}
\newcommand{\ZZ}{{\mathbb Z}}

\newcommand{\VV}{{\mathbb V}}
\newcommand{\WW}{{\mathbb W}}

\newcommand{\bpi}{\boldsymbol{\pi}}

\newcommand{\bY}{\boldsymbol{Y}}

\newcommand{\bX}{\boldsymbol{X}}

\newcommand{\bSi}{\boldsymbol{\Sigma}}
\newcommand{\bT}{\boldsymbol{T}}

\newcommand{\bP}{\boldsymbol{P}}

\newcommand{\bC}{\boldsymbol{C}}

\begin{document}

\title{Discrete Darboux system with self-consistent sources and~its~symmetric~reduction}

\author{Adam Doliwa}
\address{A. Doliwa: Faculty of Mathematics and Computer Science, University of Warmia and Mazury,
ul.~S{\l}oneczna~54, 10-710~Olsztyn, Poland}
\email{doliwa@matman.uwm.edu.pl}
\urladdr{http://wmii.uwm.edu.pl/~doliwa/}

\author{Runliang Lin}
\address{R. Lin: Department of Mathematical Sciences, Tsinghua University, Beijing 100084, P.R. China}
\email{rlin@mail.tsinghua.edu.cn}
\urladdr{http://faculty.math.tsinghua.edu.cn/~rlin/}

\author{Zhe Wang}
\address{Z. Wang: Department of Mathematical Sciences, Tsinghua University, Beijing 100084, P.R. China}
\email{zhe-wang17@mails.tsinghua.edu.cn}

%

\begin{abstract}
The discrete non-commutative Darboux system of equations with self-consistent sources is constructed, utilizing both the vectorial fundamental (binary Darboux) transformation and the method of additional independent variables. Then the symmetric reduction of discrete Darboux equations with sources is presented. In order to provide a simpler version of the resulting equations we introduce the $\tau/\sigma$ form of the (commutative) discrete Darboux system.  Our equations give, in continuous limit, the version with self-consistent sources of the classical symmetric Darboux system.
\end{abstract}
\maketitle

\section{Introduction}
The study of the soliton equations with self-consistent sources dates back to Mel'nikov~\cite{Mel,Mel1987A}, and this kind of equations is of practical importance nowadays. In physics, these equations have applications in fluid mechanics, plasma physics and solid state physics. In general, the sources make the soliton waves interact differently by changing the velocity of the waves~\cite{Lin2001,Zeng2000}. 

The well-known Kadomtsev-Petviashvili (KP) equation with self-consistent sources, which in physics represents the interaction between long waves and short waves propagating in the $x-y$ plane, has the following form~\cite{Mel}:
\begin{equation}
\partial_x(4\partial_tu - 12u\partial_xu - \partial_x^3 u) - 3\partial_y^2u = -4\sum_{j = 1}^K r_jq_j,
\end{equation}
with the source term satisfying:
\begin{equation}
\partial_y q_j = \partial_x^2q_j + 2uq_j, \qquad j = 1, \cdots, K,	
\end{equation}
\begin{equation}
\partial_y r_j = -\partial_x^2r_j - 2ur_j, \qquad j = 1, \cdots, K.	
\end{equation}
This example shows the generic structure of equations with sources, where the additional source terms are bilinear expressions formed by solutions of the linear problem (and its adjoint) of the equation without sources.
For more examples, one can refer to the review~\cite{Lin2016,Lin2013} and the references therein.

Discrete integrable systems~\cite{DSI} are essential to understanding of integrability.  Recently, the discrete KP equation~\cite{Hirota,OHTI} with self-consistent sources was proposed in \cite{Hu2006} using the "source generalization" method. The main idea of the method is to replace arbitrary constants in the solutions of an integrable equation without sources by arbitrary functions, and then construct the equation based on this new expressions. In \cite{doliwakp} new point of view on discrete KP equation with sources has been proposed (see also \cite{LinDu} for additional details). It is based on description of the flow generated by source terms as generated by vectorial binary Darboux transformations. This is a discrete analogue of the squared eigenfunction symmetry \cite{Oevel1993} approach, which was used in \cite{Liu2008a} to modify a specific flow of a given hierarchy to obtain its version with self-consistent sources. Equivalently, because for discrete integrable systems there is no essential difference between Darboux transformations and shifts in additional discrete variables~\cite{LeBen}, the flow can be interpreted as collective shift in the corresponding auxiliary discrete variables. This point of view demystifies the subject of integrable equations with sources by embedding them into the equations without sources.

In \cite{HuHuTam} the `source generation' procedure was applied to three-dimensional three-wave interaction equation, which essentialy coincides with the system (in dimension threee) introduced by Darboux~\cite{Darboux-OS} in his studies of conjugate nets. The same technique was used in \cite{GegenhashiYan} to construct discrete three-wave interaction equation with self-consistent sources. The discrete Darboux equations were first proposed on the level of non-local $\bar{\partial}$-dressing method in~\cite{BoKo}. Their simple geometric meaning in terms of lattices of planar quadrilaterals (discrete conjugate nets) which gave also multidimensional consistency of the equations was presented in~\cite{MQL}. Darboux transformations of the discrete Darboux equations together with their geometric meaning were studied first in~\cite{MDS,TQL}. 

In this paper we investigate some of problems mentioned at the end of previous work~\cite{doliwakp}: \begin{itemize}
	\item non-commutative discrete Darboux equations with sources~\cite{gaql};
		\item the self-consistent source extension of the symmetric reduction of the discrete Darboux equations (known also as the discrete CKP system)~\cite{DS-sym,Schief2003Lattice}.
\end{itemize}
Its layout is as follows. In Section~\ref{sec:DD} we recall first relevant information on non-commutative discrete Darboux equations and their fundamental (binary Darboux) transformation in vectorial form. Then, on the basis of our approach, we propose their extension with self-consistent sources and the corresponding linear problems. Next we interpret the extended Darboux equation within the formalism of additional variables. The last part of this Section is devoted to presentation of new elements of the interpretation of the fundamental transformation within the formalism. They are used in the next Section~\ref{sec:sym}, where we formulate the symmetric reduction of the discrete Darboux equations with sources. In order to write down in Section~\ref{sec:CKP} the symmetric discrete Darboux equations with sources as an extension of the $\tau$-function form of the discrete CKP equation~\cite{Kashaev,Schief2003Lattice} we have to rewrite first in an analogous $\tau/\sigma$ form the original discrete Darboux equations. Then we proceed to the symmetric reduction and its extension with self-consistent sources.
Finally, in Section~\ref{sec:cont} we take the continuous limit of the discrete equations presented in previous parts of the paper. In particular, we present the symmetric Darboux equations with self-consistent sources.

\section{Discrete non-comutative Darboux system with self-consistent sources}
\label{sec:DD}
\subsection{Discrete non-commutative Darboux equations and their fundamental transformation} \label{sec:DDT}
In this Section we recall the necessary information about the discrete Darboux system of equations~\cite{BoKo,MQL,KoSchief2,DS-sym} and its fundamental (binary Darboux) transformation~\cite{MDS,TQL,MM} in the non-commutative setting~\cite{gaql}.

Let $\VV$ be right linear space over a division ring $\DD$. Consider vector-valued functions $\bX_i \colon \ZZ^N \to \VV$, $i=1,\dots , M$, which satisfy the following linear system
\begin{equation}
\label{d-lin}
\bX_{i(j)}  = \bX_i + \bX_j Q_{ij(j)}, \qquad i\neq j,
\end{equation}
with functions $Q_{ij}\colon \ZZ^N \to \DD$, $i\neq j$, called the rotation coefficients; 
here the index in brackets denotes the forward shift in the corresponding discrete variable.
Its compatibility condition 
\begin{equation}
\label{d}
 Q_{ij(k)}  =  Q_{ij} +Q_{kj}Q_{ik(k)}, \qquad 1\leq i,j,k \quad \text{distinct},
\end{equation}
forms a system of discrete non-commutative (for non-commutative $\DD$) Darboux equations. The same system is also the compatibility condition for the adjoint linear problem 
\begin{equation}
\label{d-ad}
\bX_{j(i)}^*  =  \bX_{j}^*  + Q_{ij} \bX_{i(i)}^*,  \qquad i\neq j,
\end{equation}
where $\bX_i^* \colon \ZZ^N \to \VV^*$, $i=1,\dots , M$ are functions into the dual space of $\VV$.  
\begin{Rem}
	The discrete non-commutative Darboux system can be solved (see \cite{DS-sym} for the commutative case) for evolution of functions $ Q_{ij(j)}$
	\begin{equation} \label{eq:dD-solv}
	Q_{ij(jk)} = (1 - Q_{kj(j)}Q_{jk(k)})^{-1} \left( Q_{ij(j)} + Q_{kj(j)}Q_{ik(k)}\right) .
	\end{equation}
\end{Rem}

The  discrete non-commutative Darboux system implies existence of the potentials $\rho_i \colon \ZZ^N \to \DD$ defined by equations
\begin{equation} \label{eq:rho-constr}
\rho_{i(j)} = \rho_i ( 1 - Q_{ji(i)} Q_{ij(j)}) , \quad i\ne j \; .
\end{equation}
\begin{Rem}
In addition, when the division ring $\DD$ is commutative there exists yet another potential, the $\tau$-function, given by
\begin{equation} \label{eq:rho-tau}
\rho_i = \frac{\tau_{(i)}}{\tau}.
\end{equation}
\end{Rem}

Below we present non-commutative versions of elementary symmetries of the discrete Darboux equations and of the corresponding linear problems~\cite{DS-sym}.
\begin{Cor} \label{cor:QX-freedom}
Given solutions $Q_{ij}, \bX_i, \bX^*_j$ of equations~\eqref{d-lin}-\eqref{d-ad} and given non-vanishing functions $a_i\colon \ZZ \to \DD$ of single variables $n_i$, $i=1,\dots ,N$, then the functions $\bar{Q}_{ij}, \bar{\bX}_i, \bar{\bX}^*_j$ given by
\begin{equation}
\bar{Q}_{ij(j)} = a_j^{-1}Q_{ij(j)} a_i, \qquad \bar{\bX}_i = \bX_i a_i, \qquad
\bar{\bX}^*_{j(j)} = a_j^{-1} \bX^*_{j(j)},
\end{equation}
satisfy also equations~\eqref{d-lin}-\eqref{d-ad}. Moreover, given an additional system of  non-vanishing functions $b_i\colon \ZZ \to \DD$ of single variables $n_i$, then $\bar{\rho}_i = b_i \rho_i a_i$ are the first potentials corresponding to $\bar{Q}_{ij}$.
\end{Cor}

Following~\cite{DS-sym} let us give non-commutative versions of the "diagonal" rotation coefficients $Q_{ii}$, which can be defined by the compatible system of equations
\begin{equation} \label{eq:Qii-def}
Q_{ii(j)} = Q_{ii} + Q_{ji}Q_{ij(j)}.
\end{equation}
These coefficients $Q_{ii}$ are given up to an arbitrary additive function of single variable $n_i$.

To conclude this Section let us present basic elements~\cite{MDS,TQL} the vectorial fundamental transformation (known also as the binary Darboux transformation) of the discrete Darboux equations on the non-commutative level~\cite{gaql}.
Given the solution $\bY=(\bY_i)_{i=1}^N$ of the linear problem~\eqref{d-lin} taking values in the linear space $\WW$, and given solutions $\bY^*=(Y^*_i)_{i=1}^N$ of the adjoint linear problem~\eqref{d-ad} with values in $\WW^*$, there exist the following linear operator valued potentials $\Omega$ defined by compatible equations:
	\begin{align}
	\Omega[\bY,\bY^*]_{(i)}  & =  \Omega[\bY,\bY^*] + \bY_i\otimes \bY_{i(i)}^*, \qquad \Omega[\bY,\bY^*]\in \mathrm{L}(\WW), \\
\label{eq:O-XY*} \Omega[\bX,\bY^*]_{(i)}  & =  \Omega[\bX,\bY^*] + \bX_i\otimes \bY_{i(i)}^*, \qquad \Omega[\bX,\bY^*]\in \mathrm{L}(\WW,\VV),\\
\label{eq:O-YX*}	\Omega[\bY,\bX^*]_{(i)}  & =\Omega[\bX,\bY^*]+ \bY_i\otimes \bX_{i(i)}^* , \qquad \Omega[\bY,\bX^*]\in \mathrm{L}(\VV,\WW).
	\end{align}
If the potential $\Omega[\bY,\bY^*]$ is invertible, then
	\begin{align} \label{eq:DT-X}
	\tilde{\bX}_i & = \bX_i - \Omega[\bX,\bY^*]\Omega[\bY,\bY^*]^{-1}\bY_i, \\
	\tilde{\bX}_i^* & = \bX_i^* - \bY_i^*\Omega[\bY,\bY^*]^{-1}\Omega[\bY,\bX^*],
	\end{align}
	satisfy the linear problem and its adjoint correspondingly, with the transformed fields
	\begin{equation} \label{eq:DT-Q}
	\tilde{Q}_{ij} = Q_{ij} - \bY_j^*\Omega[\bY,\bY^*]^{-1}\bY_i.
	\end{equation}
The potential $\Omega[\bX,\bX^*]$ transforms itself according to
\begin{equation}
\Omega[\tilde{\bX},\tilde{\bX}^*] = \Omega[\bX,\bX^*] - \Omega[\bX,\bY^*]\Omega[\bY,\bY^*]^{-1} \Omega[\bY,\bX^*].
\end{equation}
The corresponding transformation rule of the first potentials reads
\begin{equation}
\label{eq:fund-vect-rho}
\tilde{\rho}_i  = \rho_i(1 + \bY^*_{i(i)}
\Omega[\bY,\bY^*]^{-1}
\bY_i).
\end{equation}
\begin{Rem}
	The $\tau$-function (for $\DD$ commutative) transforms as
	\begin{equation} \label{eq:DT-tau}
	\tilde{\tau}= \tau \det \Omega[\bY, \bY^*].
	\end{equation}
	
\end{Rem}

\subsection{The discrete non-commutative Darboux equations with sources and the corresponding linear problems} Motivated by the approach introduced in~\cite{doliwakp} we propose the discrete non-commutative Darboux equations. We interpret the fundamental transformation shift as a transition in additional variable, denoted here by $n_K$. 

\begin{Lem}
	In notation of Section~\ref{sec:DDT} the vector-valued functions $\Omega[\bY,\bY^*]^{-1}\bY_i \colon \ZZ^N \to \WW$ satisfy the linear problem~\eqref{d-lin} with the transformed fields~\eqref{eq:DT-Q}.
\end{Lem}
\begin{proof}
	By direct calculation where, to shorten the notation, we write temporarily $\Omega =\Omega[\bY,\bY^*]$:
	\begin{gather*}
	(\Omega^{-1}\bY_i)_{(j)} - \Omega^{-1}\bY_i - \Omega^{-1}\bY_j \tilde{Q}_{ij(j)} =
	\Omega^{-1}_{(j)} \bY_{i(j)} - \Omega^{-1}(\bY_i + \bY_j Q_{ij(j)}) + \Omega^{-1}\bY_j \bY_{j(j)}^* \Omega^{-1}_{(j)} \bY_{i(j)} =\\
	(\Omega^{-1}_{(j)}  - \Omega^{-1})\bY_{i(j)} + 
	\Omega^{-1} (\Omega_{(j)}-\Omega)  \Omega^{-1}_{(j)} \bY_{i(j)} = 0.
	\end{gather*} 
\end{proof}
\begin{Def} \label{def:pi}
Let us denote by $(\bpi_i)_{i=1}^N$ the solution of the linear system~\eqref{d-lin} whose fundamental transform is $(\tilde{\bpi}_i)_{i=1}^N=\left(\Omega[\bY,\bY^*]^{-1}\bY_i\right)_{i=1}^N$. 	
\end{Def} 
\begin{Def}
	The \emph{discrete non-commutative Darboux equations with sources} consist of the discrete Darboux system~\eqref{d} supplemented by additional evolution direction equation
	\begin{equation} \label{eq:Q-K}
	 Q_{ij(K)} = Q_{ij} - \bY^*_j \bpi_{i(K)},
	\end{equation}
	where the source terms $\bpi_i\colon \ZZ^{N+1}\to \WW$, $\bY^*_i \colon \ZZ^{N+1} \to \WW^*$ satisfy the linear system~\eqref{d-lin} and its adjoint~\eqref{d-ad}, respectively.
\end{Def}
\begin{Rem}
	The evolution of source terms allows for freedom in the source variable direction. This property reflects the corresponding freedom in definition of the fundamental transformation. 
\end{Rem}
\begin{Cor}
	Equivalently, instead of equation \eqref{eq:Q-K} one can take 
	\begin{equation} \label{eq:Q-Y-K}
	Q_{ij(K)} = Q_{ij} -\bY^*_j \Omega[\bY,\bY^*]^{-1}\bY_{i},
	\end{equation}
	where the $\bY_i\colon \ZZ^{N+1}\to \WW$, $\bY^*_i \colon \ZZ^{N+1} \to \WW^*$ satisfy the linear system~\eqref{d-lin} and its adjoint~\eqref{d-ad}, respectively.
\end{Cor}
\begin{Prop}
	The discrete non-commutative Darboux equations with sources are compatibility condition for the linear system \eqref{d-lin} supplemented by the additional evolution direction equations
	\begin{equation} \label{eq:lin-X-K}
 \bX_{i(K)} = \bX_i - \Omega[\bX,\bY^*]\bpi_{i(K)},
	\end{equation}
	where the $\mathrm{L}(\VV,\WW)$-valued potential $\Omega[\bX,\bY^*]$ satisfies equation~\eqref{eq:O-XY*}.
\end{Prop}
\begin{proof}
The compatibility condition of equations~\eqref{eq:O-XY*} and the linear system~\eqref{d-lin} imply that $(\bY^*_i)_{i=1}^N$ satisfy the adjoint linear problem~\eqref{d-ad}. Then the compatibility of the ordinary part~\eqref{d-lin} of the linear problem within the set of ordinary variables gives the ordinary part~\eqref{d} of the non-linear system. The compatibility of the ordinary part of the linear system with the additional variable evolution part~\eqref{eq:lin-X-K} gives 
\begin{equation*}
(\bX_j - \Omega[\bX,\bY^*]\bpi_{j(K)}) Q_{ij(jK)} - \Omega[\bX,\bY^*] \bpi_{i(K)} =
\bX_j Q_{ij(j)}  - (\Omega[\bX,\bY^*] + \bX_j\otimes \bY_{j(j)})^*)\bpi_{i(jK)} .
\end{equation*}
By collecting coefficients at $\Omega[\bX,\bY^*]$ we obtain the linear equations for $(\bpi_i)_{i=1}^N$, while the coefficients at $\bX_j$ produce the source part evolution equations~\eqref{eq:Q-K}.
\end{proof}
\begin{Rem}
	Notice that the form of the additional linear equation \eqref{eq:lin-X-K} corresponds directly to the transformation formulae~\eqref{eq:DT-X}.
\end{Rem}
We leave similar proofs of the following two Corollaries to the interested Reader. In the next Section we will prove them in another formalism of interpretation of the equations with sources.
\begin{Cor}
	The discrete non-commutative Darboux equations with sources are compatibility condition for the adjoint linear system \eqref{d-ad} supplemented by the additional evolution direction equation
	\begin{equation} \label{eq:lin-X*-K}
	 \bX_{i(K)}^* = \bX_i^*- \bY^*_i \Omega[\bpi,\bX^*]_{(K)},
	\end{equation}
	where the $\mathrm{L}(\WW,\VV)$-valued potential $\Omega[\bpi,\bX^*]$ satisfies equation \eqref{eq:O-YX*} with $\bpi_i$ in the place of $\bY_i$.
\end{Cor}

\begin{Prop}
	The following evolution rule of the first potentials $\rho_i$ in the additional direction
	\begin{equation} \label{eq:rho-K}
\rho_{i(K)}  = \rho_i(1 + \bY^*_{i(i)} \bpi_{i(K)}),
	\end{equation}
is compatible with the discrete non-commutative Darboux equations with sources, and is the direct counterpart of equations	\eqref{eq:fund-vect-rho}.	
\end{Prop}
\begin{proof}
	By direct calculation using the identities
\begin{align}
\label{id-1}
(1+\bY_{j(j)}^*\bpi_{j(K)})Q_{ij(jK)} & = Q_{ij(j)} - \bY_{j(j)}^*\bpi_{i(K)},\\
\label{id-2}
(1-Q_{ji(i)}Q_{ij(j)}) \bY_{i(ij)}^*&= \bY_{i(i)}^* + Q_{ji(i)}\bY_{j(j)}^*.
\end{align}
\end{proof}
\begin{Rem}
	Notice the following identity
\begin{equation}
(1+\bY_{i(i)}^*\bpi_{i(K)})\bX^*_{i(iK)} = \bX^*_{i(i)} - \bY_{i(i)}^* \Omega[\bpi,\bX^*]_{(K)},
\end{equation}
which on the level of the fundamental transformation reads
\begin{equation} \label{eq:id-Yi}
(1 + \bY^*_{i(i)}\Omega[\bY,\bY^*]^{-1}\bY_i)  \tilde{\bX}^*_{i(i)} = \bX^*_{i(i)} 
- \bY_{i(i)}^* \Omega[\bY,\bY^*]^{-1}\Omega[\bY,\bX^*].
\end{equation}
\end{Rem}

\subsection{Discrete Darboux equations with sources in the formalism of additional variables}
In this Section, we are going to construct the discrete Darboux equation with  self-consistent sources using the second method. The direct motivation is based on the observation~\cite{TQL} that the scalar fundamental transformation ($\dim \WW = 1$) cannot be distinguished from the shift in a discrete variable. This corresponds also with the well known method of integrable discretization by using Darboux--B\"{a}cklund transformations~\cite{LeBen}.

Let us split independent variables of the discrete Darboux equations into disjoint sets of \emph{ordinary variables} $n_i$, $i=1,\dots , N$, indexed by $i,j,k, \dots$, and the \emph{source variables} $m_a$, $a=1,\dots , K$, indexed by $a,b,c,\dots$. This way of indexing applies also to corresponding indices of solutions of the linear problem and its adjoint, and indices of the dependent variables. 
\begin{Prop} \label{prop:Y*-pi}
	Denote by $\bY^*_j = (Q_{aj})_{a=1}^K$ the row-vector valued function build up of the solution of the full system of discrete Darboux equations, and by $\bpi_i = -(Q_{ia})_{a=1}^K$ the similar column-vector valued function, then the corresponding parts of the Darboux equations
	\begin{align}
	Q_{aj(i)} & = Q_{aj} + Q_{ij} Q_{ai(i)},\\
	Q_{ia(j)} & = Q_{ia} + Q_{ja} Q_{ij(j)},
	\end{align}
can be interpreted as the linear systems
\begin{align}
\bY^*_{j(i)} & = \bY^*_{j} + Q_{ij}\bY^*_{i(i)},\\
\bpi_{i(j)} & = \bpi_{i} + \bpi_{j} Q_{ij(j)},
\end{align}
respectively.
\end{Prop}
The above proposition gives the interpretation of the source terms in the new formalism. The source evolution direction is then the collective shift in all additional variables, and the corresponding part of the system \eqref{eq:Q-K} is provided by the corresponding interpreation of the following result, first found in the commutative case in~\cite{DS-sym}. 
\begin{Prop} \label{prop:Q-K}
	Denote by $(K)$ the simultaneous shift in the variables $(m_1,\dots , m_K)$ then 
	\begin{equation} \label{eq:Qij-K}
	Q_{ij(K)} = Q_{ij} + \sum_a Q_{aj} Q_{ia(K)}.
	\end{equation}
\end{Prop}
\begin{proof}
	The proof is by mathematical induction with respect to $K$. For $K=1$ we have just the corresponding part of the discrete Darboux system. The induction step reads
	\begin{equation*}
	Q_{ij(Kb)}   =
	Q_{ij} +  \sum_a Q_{aj}  Q_{ia(Kb)} + Q_{bj} \left(Q_{ib} + \sum_a  Q_{aj}  Q_{ia(K)} \right)_{(b)}.
	\end{equation*}
\end{proof}

Let us present the corresponding derivation of the additional equations \eqref{eq:lin-X-K} and \eqref{eq:lin-X*-K} of the linear problem and of its adjoint. In the same way as above one can show the following equations which express simultaneous shifts in the additional variables of the ordinary vectors $\bX_i$ and $\bX^*_j$
\begin{align}
\bX_{i(K)} & = \bX_i + \sum_a \bX_a Q_{ia(K)},\\
\bX^*_{j(K)} & = \bX_j^* + \sum_a Q_{aj} \bX^*_{a(K)}.
\end{align}
To conclude alternative derivation of the linear problem~\eqref{eq:lin-X-K} and its adjoint~\eqref{eq:lin-X*-K} it is enough to observe that the sums in the above equations can be interpreted as product of certain matrix-valued potentials with the vectors $\bpi_{i(K)}$ and $\bY^*_j$, respectively.
\begin{Lem}
	 The matrix $(\bX_a)_{a=1}^K$ composed of the additional column-vectors, can be interpreted as the potential $\Omega[\bX,\bY^*]$, while the matrix $(\bX^*_a)_{a=1}^K$ composed of the additional row-vectors can be interpreted as the potential $-\Omega[\bpi,\bX^*]$.	
\end{Lem}
\begin{proof}
	The interpretation follows from the corresponding parts of the linear problem~\eqref{d-lin} and its adjoint~\eqref{d-ad}:
	\begin{align*}
	\bX_{a(i)} & = \bX_a + \bX_i Q_{ai(i)},\\
	\bX_{a(i)}^* & = \bX^*_a + Q_{ia} \bX^*_{i(i)}.
	\end{align*}
\end{proof}
\begin{Cor}
	Equation \eqref{eq:rho-K} in the present formalism is just another form of the equation
	\begin{equation*}
	\rho_{i(K)} = \rho_i \left( 1 - \sum_a Q_{ai(i)} Q_{ia(K)}\right),
	\end{equation*}
	which can be proved by induction, in a similar way like Proposition~\ref{prop:Q-K}.
\end{Cor}

\begin{Rem}
	Let us select one of additional variables, and denote temporarily its index by $\ell$. Moreover, up to renumeration, let us write $\bY_j^* = (Q_{\ell j}, \check{\bY}^*_j)$ and $\bpi_i = (-Q_{i\ell}, \check{\bpi}_i)$. Then the additional evolution direction equation \eqref{eq:Q-K} takes the form
	\begin{equation}
	Q_{ij(K)} = Q_{ij} + Q_{\ell j} Q_{i\ell (K)} - \check{\bY}^*_j \check{\bpi}_{i(K)},
	\end{equation}
which agrees with the corresponding equations derived in~\cite{GegenhashiYan} as \emph{discrete three-dimensional three wave interaction equations with self-consistent sources} using the so called source generalization method (in the commutative case). 
The corresponding additional parts of the linear problem and its adjoint read
\begin{align} 
 \bX_{i(K)} & = \bX_{i} + \bX_{\ell} Q_{i\ell(K)} - \Omega[\bX,\check{\bY}^*]\check{\bpi}_{i(K)},\\
\bX^*_{j(K)} & = \bX^*_{j} + Q_{\ell j} \bX^*_{\ell(K)} - \check{\bY}^*_j \Omega[\check{\bpi},\bX^*].
\end{align}
This form directly reduces to the discrete Darboux equations in the absence of sources. Our form is simpler and will be used in making further reductions, as presented in the Section~\ref{sec:sym}.
\end{Rem}

\subsection{The fundamental transformation within the additional variables approach}
For discrete Darboux equations the correspondence between fundamental transformations and shifts in additional variables was observed first in~\cite{TQL} mainly on the geometric level, but the direct vocabulary was not given there. Proposition~\ref{prop:Y*-pi}
presents the relevant description of the solutions $\bY^*_i=(Q_{ai})_{a=1}^K$ of the adjoint linear problem~\eqref{d-ad}. Let us complete the vocabulary by presenting the corresponding meaning of the solutions $\bY_i$ of the linear problem, which enter in the definition of the fundamental transformation. This will give better understanding of the discrete Darboux equations, and will be used in description of the symmetric reduction of the equations, see the next Section. 

We start with presenting such an interpretation of functions $Q_{ab}$, whose both indices are of the source-variable type.
\begin{Prop}
	The matrix $(Q_{ab})_{a,b=1}^K$ can be interpreted as the potential $-\Omega[\bpi,\bY^*]$.
\end{Prop}
\begin{proof}
	It is enough to write down the corresponding part of the discrete Darboux equations~\eqref{d} supplemented by the definition~\eqref{eq:Qii-def} of the diagonal rotation coefficients
	\begin{equation*}
	Q_{ab(i)} = Q_{ab} + Q_{ib} Q_{ai(i)} = Q_{ab} - \left( \bpi_i\otimes \bY^*_{i(i)}\right)_{ab},
	\end{equation*}
	and compare with the definition of  matrix potentials.
\end{proof}

\begin{Prop} \label{prop:Y-O}
 The column-vector valued functions $\bY_i = (\rho_a Q_{ia(a)})_{a=1}^K$ build up of the solutions of the full system of discrete Darboux equations satisfy the linear problem \eqref{d-lin}. Together with the solutions $\bY_i^*=(Q_{ai})_{a=1}^K$ of the adjoint linear problem \eqref{d-ad} discussed in Proposition~\ref{prop:Y*-pi}, they define the vectorial fundamental transformation interpreted in Proposition~\ref{prop:Q-K} as the shift in all additional variables. Moreover, the matrix elements of the potential $\Omega[\bY,\bY^*]$
 can be identified with
 \begin{equation}
 \Omega[\bY,\bY^*]_{ab} = \begin{cases}
 \rho_b Q_{ab(b)} & \text{for} \quad a\neq b,\\
 -\rho_a & \text{for} \quad a=b . \end{cases}
 \end{equation}
\end{Prop}
\begin{proof}
	To show that the functions $\bY_i$, as defined above, satisfy the linear problem \eqref{d-lin} we use definition~\eqref{eq:rho-constr} of the potentials $\rho_i$ and the identity~\eqref{eq:dD-solv} in the form
	\begin{equation*}
	(1-Q_{ja(a)}Q_{aj(j)})Q_{ia(aj)} = Q_{ia(a)} + Q_{ja(a)} Q_{ij(j)}.
	\end{equation*}	
	Analogous reasoning applied to potentials $\rho_b$ and the identity~\eqref{eq:dD-solv} in the form	
	\begin{equation*}
	(1-Q_{ib(b)}Q_{bi(i)})Q_{ab(bi)} = Q_{ab(b)} + Q_{ib(b)} Q_{ai(i)},
	\end{equation*}
	provide the interpretation of the off-diagonal elements of the potential $\Omega[\bY,\bY^*]$. The identification  of the diagonal elements follows from definition \eqref{eq:rho-constr} of the potentials $\rho_a$ written in the form
	\begin{equation*}
	-\rho_{a(i)} = -\rho_a + \rho_a Q_{ia(a)}Q_{ai(i)}.
	\end{equation*}

	To conclude the proof we need to show that the solutions $\bpi_i = (-Q_{ia})_{a=1}^K$ and  $\bY_i = (\rho_a Q_{ia(a)})_{a=1}^K$ of the linear problem~\eqref{d-lin} are matched by
	\begin{equation} \label{eq:Om-pi-Y}
	\Omega [\bY,\bY^*] \bpi_{i(K)} = \bY_i,
	\end{equation}
	in agreement with Definition~\ref{def:pi}. Denote by $(\check{K}^a)$ the simultaneous shift in the additional variables with $m_a$ excluded. Then by application of equation \eqref{eq:Qij-K} for index $j$ replaced by $a$, and the shift $(K)$ replaced by $(\check{K}^a)$ we obtain 
	\begin{equation*}
	Q_{ia(\check{K}^a)} = Q_{ia} + \sum_{b\in \check{K}^a} Q_{ba} Q_{i b (\check{K}^a)} .
	\end{equation*}
	By shifting the above relation in direction of $m_a$ and by multiplying its both sides by $\rho_a$ we get
	\begin{equation*}
	\rho_a Q_{ia(K)} - \sum_{b\in \check{K}^a} \rho_a Q_{ba(a)} Q_{i b (K)} = \rho_a Q_{ia(a)},
	\end{equation*}
which gives the matching condition~\eqref{eq:Om-pi-Y}.
\end{proof}

\section{The symmetric reduction of the discrete Darboux system with sources} \label{sec:sym}
From now on we assume \emph{commutativity} of the division ring $\DD$ (the geometric reason for such a restriction in the case of the symmetric discrete Darboux equations was given in~\cite{gaql}). Our goal is to introduce and study properties of the symmetric reduction of the discrete Darboux equations with self-consistent sources. 
\subsection{The symmetric discrete Darboux system and its fundamental transformation}
The symmetric discrete Darboux system, called also the discrete CKP equations, was introduced in~\cite{DS-sym,Schief2003Lattice} and studied subsequently in~\cite{KingSchief,CQL}. The corresponding reduction of the fundamental transformation was given in~\cite{MM}. The $\tau$-function formulation of the symmetric discrete Darboux equations was identified in~\cite{Schief2003Lattice} as an identity satisfied by the B\"{a}cklund transformations of the CKP hierarchy. It turns out that the same equation was introduced in a different context in~\cite{Kashaev}. The present Section is devoted to presentation of the basic results of the subject.

\begin{Def}{\cite{DS-sym}}
	The discrete Darboux equations with the constraint
\begin{equation} \label{eq:C-contr}
\rho_iQ_{ji(i)} = \rho_jQ_{ij(j)},
\end{equation}	
are called the symmetric discrete Darboux system.
\end{Def}
\begin{Rem}
	Up to allowed freedom described in Corollary~\ref{cor:QX-freedom} the above constraint is equivalent to the following relation
	\begin{equation}
	Q_{ij(j)}Q_{jk(k)} Q_{ki(i))} = Q_{ji(i)} Q_{kj(j)} Q_{ik(k)}.
	\end{equation}
\end{Rem}
\begin{Cor}
The symmetric constraint~\eqref{eq:C-contr} in view of equations~\eqref{eq:rho-constr}-\eqref{eq:rho-tau} can be formulated in the form
\begin{equation}
(\tau_{(j)} Q_{ij(j)})^2 = \tau_{(i)} \tau_{(j)} - \tau \tau_{(ij)},
\end{equation}
which inserted into the discrete Darboux system gives rise to the following formulation of the discrete CKP equations~\cite{Kashaev,Schief2003Lattice}
\begin{equation}
\label{taud}
\begin{split}
&(\tau\tau_{(ijk)} - \tau_{(i)}\tau_{(jk)} - \tau_{(j)}\tau_{(ik)} - \tau_{(k)}\tau_{(ij)})^2 + 4(\tau_{(i)}\tau_{(j)}\tau_{(k)}\tau_{(ijk)} + \tau\tau_{(ij)}\tau_{(jk)}\tau_{(ik)})  \\
& \qquad \qquad = 4(\tau_{(i)}\tau_{(jk)}\tau_{(j)}\tau_{(ik)} + \tau_{(i)}\tau_{(jk)}\tau_{(k)}\tau_{(ij)} + \tau_{(j)}\tau_{(ik)}\tau_{(k)}\tau_{(ij)}).
\end{split}
\end{equation}
\end{Cor}
\begin{Rem}
	Because of the sign ambiguity in solving the above system for $\tau_{(ijk)}$ it cannot be considered as a legitimate recurrence~\cite{TsarevWolf}. A way to make from \eqref{taud} a single-valued system (essentially going back to the rotation coefficients) was proposed in~\cite{Atkinson}, and will be discussed in the next Section. 
\end{Rem}
\begin{Rem}
	It is interesting to notice~\cite{TsarevWolf} that analogous equations are satisfied by principal minors of symmetric matrices~\cite{HoltzSturmfels}.
\end{Rem}
\begin{Rem}
	The discrete CKP equation~\eqref{taud} with $N$-independent variables can be obtained by symmetry reduction from Hirota's discrete AKP equation~\cite{Hirota} with $2N-1$ independent variables. Details, which involve description of the latter system in terms of $A_{2N-1}$ root lattice~\cite{Dol-AN} can be found in~\cite{Dol-Tampa}. 
\end{Rem}

Construction of the corresponding reduction of the fundamental transformation~\cite{MM} makes use of the result given in \cite{DS-sym}.
\begin{Th} \label{th:C-char}
	The following three characterizations of the symmetric reduction of the discrete Darboux equations are equivalent:
\begin{enumerate}
	\item The original symmetric reduction constraint~\eqref{eq:C-contr}.
	\item Given a non-trivial solution $\bX_i^*$ of the adjoint linear problem~\eqref{d-ad} then 
	\begin{equation} \label{eq:X-X*}
	\bX_i = \rho_i(\bX_{i(i)}^*)^t
	\end{equation}
	provides a solution of the linear problem~\eqref{d-lin}.
	\item The matrix valued potential $\Omega[\bX,\bX^*]\in \mathrm{L}(\VV)$ is symmetric
\begin{equation} \label{eq:O-Ot}
\Omega[\bX,\bX^*]^t = \Omega[\bX,\bX^*],
\end{equation}
 (up to choice of the initial condition).
\end{enumerate}
\end{Th}
When data $(\bY_i)_{i=1}^N$, $(\bY^*_i)_{i=1}^N$ and $\Omega[\bY,\bY^*]$ of the fundamental transformation of the symmetric Darboux equations have been chosen according the constraint~\eqref{eq:X-X*}-\eqref{eq:O-Ot} then the transformed solutions $\tilde{Q}_{ij}$ of the discrete Darboux system, given by \eqref{eq:DT-Q}, and the transformed potentials~$\tilde{\rho}_i$, given by \eqref{eq:fund-vect-rho}, satisfy the constraint~\eqref{eq:C-contr} as well.
\begin{Cor}
	In order the constraint \eqref{eq:X-X*} be satisfied also on the level of the transformed linear problems then the following  constraint between the transformation potentials
	\begin{equation}
	\Omega[\bX,\bY^*]^t = \Omega[\bY,\bX^*],
	\end{equation}
	which is automatically preserved up to the choice of initial data, should be imposed as well.
\end{Cor}
\begin{proof}
	The transformation formulas~\eqref{eq:DT-X}-\eqref{eq:fund-vect-rho} and definitions of the potentials imply that
	\begin{equation*}
	\tilde{\rho}_i\tilde{\bX}^*_{i(i)} - \tilde{\bX}^t_i = \rho_i \bX^*_{i(i)} - \bX^t_i - \rho_i \bY^*_{i(i)} \Omega[\bY,\bY^*]^{-1} \Omega[\bY,\bX^*] + \bY^t_i (\Omega[\bY,\bY^*]^t)^{-1} \Omega[\bX,\bY^*]^t, 
	\end{equation*}
	\begin{equation*}
	(\Omega[\bX,\bY^*]^t - \Omega[\bY,\bX^*])_{(i)} = \Omega[\bX,\bY^*]^t - \Omega[\bY,\bX^*] + \bY^{*t}_{i(i)} \otimes \bX^t_i - \bY_i \otimes \bX^*_{i(i)},
	\end{equation*}
	what gives the statement.
\end{proof}
\begin{Rem}
	In notation of Proposition~\ref{prop:Y-O} constraints~\eqref{eq:X-X*} and \eqref{eq:O-Ot} imposed on the data of the fundamental transformation repeat the basic constraint~\eqref{eq:C-contr}, but on the level of additional variables. Equation~\eqref{eq:X-X*} reads then $\rho_a Q_{ia(a)} = \rho_i Q_{ai(i)}$, and equation~\eqref{eq:O-Ot} reads $\rho_b Q_{ab(b)} = \rho_a Q_{ba(a)}$.
\end{Rem}

\subsection{The symmetric discrete Darboux system with self-consistent sources}
By putting together the description of discrete Darboux equations with sources, the symmetric Darboux equations and the corresponding reduction of the fundamental transformation we propose as follows.
\begin{Def} \label{def:sDs}
	The symmetric discrete Darboux equations with self-consistent sources is the following system consisting of:
	\begin{enumerate}
		\item discrete Darboux equations \eqref{d} and definition of the first potentials~\eqref{eq:rho-constr} subject to the symmetric constraint~\eqref{eq:C-contr};
		\item the additional part of the evolution governed by self-consistent sources~\eqref{eq:Q-Y-K} with solutions of the linear and adjoint linear problems
		constrained by \eqref{eq:X-X*} and \eqref{eq:O-Ot}.
	\end{enumerate}
\end{Def}
Integrability of the symmetric constraints on discrete Darboux equations with sources follows from the corresponding result on reduction of the fundamental transformation. Let us provide an alternative proof of their integrability using the method applied in~\cite{DS-sym} to prove the analogous result for discrete Darboux equations without sources.  
\begin{Prop} \label{prop:symetr-prop}
	The above constraints imposed on solutions of the discrete Darboux equations with sources are preserved during the evolution.
\end{Prop}
\begin{proof}
Define on the level of discrete Darboux equations with sources the following scalar functions
\begin{equation*}
C_{ij} = \rho_i Q_{ji(i)} - \rho_j Q_{ij(j)},
\end{equation*}
row-vector valued functions
\begin{equation*}
\bC_i = \bY_i^t - \rho_i \bY^*_{i(i)},
\end{equation*}
and the matrix valued function
\begin{equation*}
\bC = \Omega[\bY,\bY^*] - \Omega[\bY,\bY^*]^t.
\end{equation*}
Using the discrete Darboux equations with sources and their consequences~\eqref{eq:rho-K}-\eqref{id-2} one shows that
\begin{align*}
C_{ij(k)} & = C_{ij} + Q_{jk(k)} C_{ik} - Q_{ik(k)} C_{jk},\\
C_{ij(K)} & = C_{ij} + \bC_i \bpi_{j(K)} - \bC_j \bpi_{i(K)} - \bpi_{i(K)}^t \bC \bpi_{i(K)}, \\
\bC_{i(j)} & = \bC_i + Q_{ij(j)} \bC_j - \bY^*_{j(j)} C_{ij},\\
\bC_{(i)} & = \bC - \bY^{*t}_{i(i)} \bC_i + \bC^t_i \bY^*_{i(i)}.
\end{align*} 
In particular, we conclude that the constraints propagate once satisfied on the initial data.  
\end{proof}
The above form of the symmetric discrete Darboux equations with sources involves, apart from the source terms, the original fields $Q_{ij}$ and the potentials $\rho_i$. Therefore one can ask natural question about their formulation involving the $\tau$-function form~\eqref{taud} of the equations without sources. In order to do so we need first to rewrite both the discrete Darboux equations (for commutative $\DD$) and the fundamental transformation formulas using new, more suitable for our goal, functions. 

\section{The $\tau/\sigma$ form of the symmetric discrete Darboux equations with sources}
\label{sec:CKP}
\subsection{The $\tau/\sigma$ form of the discrete Darboux equations and of the fundamental transformation}
 Let us introduce the functions (compare with the corresponding $\tau$-function form of the discrete Darboux equations in~\cite{DMMMS,DS-sym})
\begin{equation} \label{eq:s-t-Q}
\sigma_{ij} = \tau_{(j)} Q_{ij(j)},
\end{equation}
then the discrete Darboux equations (in fact their solved form \eqref{eq:dD-solv}) and equations \eqref{eq:rho-constr}-\eqref{eq:rho-tau} can be rewritten as
\begin{align} \label{eq:s-ijk}
\tau \sigma_{ij(k)} & = \tau_{(k)} \sigma_{ij} + \sigma_{ik}\sigma_{kj}, \\
\label{eq:s-ij}
\sigma_{ij} \sigma_{ji} & = \tau_{(i)} \tau_{(j)} - \tau\tau_{(ij)},
\end{align}
which we call their  $\tau/\sigma$ form.
In order to provide the corresponding linear problems we define
\begin{equation}
\label{eq:Sigma-X}
\bSi_i = \frac{1}{\tau_{(i)}} \bX_i, \qquad \bSi_i^* = \tau_{(i)} \bX^*_{i(i)},
\end{equation}
then equations \eqref{d-lin} and \eqref{d-ad} read, respectively, 
\begin{align} \label{eq:Sij}
\tau_{(ij)}\bSi_{i(j)} & = \tau_{(i)} \bSi_i + \sigma_{ij} \bSi_j ,\\
\label{eq:Sij*}
\tau\bSi_{i(j)}^* & = \tau_{(j)} \bSi_i^* + \sigma_{ji} \bSi_j^* .
\end{align}
\begin{Rem}
	Notice that equation~\eqref{eq:Sij*} is, in fact,  direct consequence of the identity~\eqref{id-2} for $\bX^*_i$.
\end{Rem}
Moreover, abusing slightly the notation, we define the matrix potential $\Omega[\bSi,\bSi^*]$ by the compatible system
\begin{equation} \label{eq:O-SS*}
\Omega[\bSi,\bSi^*]_{(i)} = \Omega[\bSi,\bSi^*] + \bSi_i \otimes \bSi^*_i.
\end{equation}
Notice that, by equations~\eqref{eq:Sigma-X}, the potential $\Omega[\bSi,\bSi^*]$ indeed equals to $\Omega[\bX,\bX^*]$, what justifies the abuse. We will transfer it to other solutions of the linear problems \eqref{eq:Sij}-\eqref{eq:Sij*}.
\begin{Prop} \label{prop:lpD-t/s}
	Equations \eqref{eq:s-ijk}-\eqref{eq:s-ij} are the compatibility condition of the system \eqref{eq:Sij}-\eqref{eq:O-SS*}.
\end{Prop}
\begin{proof}
Compatibility of equations \eqref{eq:O-SS*} and application of the linear systems \eqref{eq:Sij} and its adjoint \eqref{eq:Sij*} gives equations~\eqref{eq:s-ij}. Moreover
consistency  of the first linear system \eqref{eq:Sij} gives
\begin{equation} \label{eq:s-ijk-t}
(\tau_{(j)} \tau_{(k)} - \sigma_{jk} \sigma_{kj}) \sigma_{ij(k)} = \tau_{(jk)} \tau_{(k)} \sigma_{ij} + \tau_{(jk)} \sigma_{ik} \sigma_{kj},
\end{equation}	
which in view of~\eqref{eq:s-ij} leads to the first part \eqref{eq:s-ijk} of the $\tau/\sigma$ form of the discrete Darboux equations. 
\end{proof}
\begin{Rem}
	The same reasoning can be applied if we used consistency of the adjoint linear problem~\eqref{eq:Sij*} in the place of \eqref{eq:Sij}.
\end{Rem}
\begin{Rem}
	Contrary to the standard form~\eqref{d} of the discrete Darboux equations the linear problem \eqref{eq:Sij} alone is not enough to write down the full $\tau/\sigma$ system.
\end{Rem}

To conclude this part let us state the corresponding result on the fundamental transformation of the $\tau/\sigma$ discrete Darboux equations. We use definitions~\eqref{eq:s-t-Q},~\eqref{eq:S-S*}, basic transformation formulas \eqref{eq:DT-X}-\eqref{eq:DT-tau} and their direct consequences.
\begin{Prop} \label{prop:fT-t/s}
	Given solution of the $\tau/\sigma$ discrete Darboux equations~\eqref{eq:s-ijk}-\eqref{eq:s-ij} and the corresponding solution of the linear system~\eqref{eq:Sij}-\eqref{eq:O-SS*}, and given additional solutions $(\bT_i)$ and $(\bT_i^*)$ of equations \eqref{eq:Sij} and \eqref{eq:Sij*}, respectively, taking values in a vector space and its adjoint. Then the transformed solution of the $\tau/\sigma$ discrete Darboux equations~\eqref{eq:s-ijk}-\eqref{eq:s-ij} and the corresponding transformed solution of the linear system~\eqref{eq:Sij}-\eqref{eq:O-SS*} read
	\begin{align} \label{eq:T-tau}
	\tilde{\tau} & = \tau \; \det \Omega[\bT,\bT^*],\\ \label{eq:T-sigma}
	\tilde{\sigma}_{ij}& =\det \Omega[\bT,\bT^*]\left( \sigma_{ij} - \bT^*_j \, \Omega[\bT,\bT^*]^{-1} \, \bT_i \tau_{(i)} \right) ,\\ \label{eq:T-Si}
	\tilde{\bSi}_i & =\left( \bSi_i - \Omega[\bSi,\bT^*]  \; \Omega[\bT,\bT^*]^{-1} \bT_i \right) / \det \Omega[\bT,\bT^*]_{(i)},\\ \label{eq:T-Si*}
	\tilde{\bSi}^*_i & = \det \Omega[\bT,\bT^*] \left( \bSi^*_i - \bT^*_i \, \Omega[\bT,\bT^*]^{-1} \Omega[\bT,\bSi^*] \right), \\
	\Omega[\tilde{\bSi},\tilde{\bSi}^*] & = \Omega[\bSi,\bSi^*] - \Omega[\bSi,\bT^*]\, \Omega[\bT,\bT^*]^{-1} \Omega[\bT,\bSi^*]. \label{eq:T-O}
	\end{align}
\end{Prop}
\begin{proof}
	Transformation formulas \eqref{eq:T-tau}, \eqref{eq:T-Si} and \eqref{eq:T-O} follow directly from the original corresponding transformation rules. To show \eqref{eq:T-sigma} we write
\begin{equation*}
\tilde{\sigma}_{ij} = \tilde{\tau}_{(j)} \tilde{Q}_{ij(j)} = \tau_{(j)} \det \Omega[\bY,\bY^*] \left( 1 + \bY^*_{j(j)} \Omega[\bY,\bY^*]^{-1} \bY_j \right)
\left( Q_{ij(j)} - \bY^*_{j(j)} \Omega[\bY,\bY^*]^{-1}_{(j)} \bY_{i(j)} \right), \end{equation*}
which simplifies to
\begin{equation*}
\tau_{(j)} \det \Omega[\bY,\bY^*] 
\left( Q_{ij(j)} - \bY^*_{j(j)} \Omega[\bY,\bY^*]^{-1} \bY_i \right) = \det \Omega[\bT,\bT^*] \left( \sigma_{ij} - \bT^*_j \, \Omega[\bT,\bT^*]^{-1} \, \bT_i \tau_{(i)}\right) .
\end{equation*}
Transformation rule~\eqref{eq:T-Si*} can be shown similarly.
\end{proof}
\begin{Rem}
	In showing directly that the transformed functions satisfy indeed the correct equations it may be convenient to use the non-reduced expressions
	\begin{align}
	\label{eq:T-sigma-n}
	\tilde{\sigma}_{ij} & = \det \Omega[\bT,\bT^*]_{(j)} \left( \sigma_{ij} - \bT^*_j \, \Omega[\bT,\bT^*]_{(j)}^{-1}\, \bT_{i(j)} \tau_{(ij)} \right) , \\
	\tilde{\bSi}_i & =\left( \bSi_i - \Omega[\bSi,\bT^*]_{(i)} \; \Omega[\bT,\bT^*]^{-1}_{(i)} \bT_i \right) / \det \Omega[\bT,\bT^*],\\
	\tilde{\bSi}^*_{i} & = \det \Omega[\bT,\bT^*]_{(i)} \left( \bSi^*_i - \bT^*_i \, \Omega[\bT,\bT^*]_{(i)}^{-1}\, \Omega[\bSi,\bSi^*]_{(i)} \right).
	\end{align} 
\end{Rem}
Again, one can interpret the vectorial fundamental transformation as shift in additional evolution direction governed by the sources $\bT_i$ and $\bT^*_i$. We leave the details for the interested Reader. 
In order to formulate $\tau/\sigma$ form of the discrete Darboux equations with sources in full analogy with the previous case we need the analogue of the solution $(\bpi_i)$ of the linear problem \eqref{d-lin}. Using Definition~\ref{def:pi} and the first of equations \eqref{eq:Sigma-X} by $(\bP_i)$ denote the solution of the linear problem~\eqref{eq:Sij} whose fundamental transform is 
\begin{equation}
\tilde{\bP}_i= \Omega[\bT,\bT^*]^{-1} \bT_i \, / \det\Omega[\bT,\bT^*]_{(i)}.
\end{equation}
One can verify directly that indeed $(\tilde{\bP}_i)$ satisfy the linear problem~\eqref{eq:Sij} with functions $\tilde{\tau}$ and $\tilde{\sigma}_{ij}$ transformed according to \eqref{eq:T-tau} and~\eqref{eq:T-sigma}/\eqref{eq:T-sigma-n}. 
\begin{Rem}
	One can transfer the additional variables interpretation also to the $\tau/\sigma$ form. Then the data of the transformation read
\begin{equation}
\bT_i = \frac{1}{\tau \tau_{(i)}}\left( \sigma_{ia} \right)_{a=1}^K , \qquad \bT^*_{i} = \left( \sigma_{ai} \right)_{a=1}^K, 
\end{equation} 
and
\begin{equation}
\Omega[\bT,\bT^*]_{ab} = \begin{cases}
{\sigma_{ab}}/{\tau} & \text{for} \quad a\neq b,\\
-{\tau_{(a)}}/{\tau} & \text{for} \quad a=b . \end{cases}
\end{equation}
\end{Rem}

\subsection{The $\tau/\sigma$ form of the symmetric discrete Darboux equations}
The symmetric reduction condition~\eqref{eq:C-contr} in the $\tau/\sigma$ notation takes the folowing simple form 
\begin{equation} \label{eq:C-t/s}
\sigma_{ij}= \sigma_{ji},
\end{equation}
and due to equation \eqref{eq:s-ij} leads to
\begin{equation} \label{eq:stt}
\sigma_{ij}^2 = \tau_{(i)} \tau_{(j)} - \tau\tau_{(ij)}.
\end{equation}

Let us formulate the $\tau/\sigma$ analog of Theorem~\ref{th:C-char}.
\begin{Prop} \label{prop:C-char-t/s}
The following three characterizations of the symmetric reduction of the discrete Darboux equations in the $\tau/\sigma$ form are equivalent:
\begin{enumerate}
	\item The original symmetric reduction constraint~\eqref{eq:C-t/s}.
	\item Given a non-trivial solution $\bSi_i$ of the linear problem~\eqref{eq:Sij} then 
	\begin{equation} \label{eq:S-S*}
	\bSi_i^* = \tau \tau_{(i)} \bSi_i^t
	\end{equation}
	provides a solution of the adjoint linear problem~\eqref{eq:Sij*}.
	\item The matrix valued potential $\Omega[\bSi,\bSi^*]$ is symmetric
	\begin{equation} \label{eq:O-Ot-t/s}
	\Omega[\bSi,\bSi^*]^t = \Omega[\bSi,\bSi^*],
	\end{equation}
	(up to choice of the initial condition).
\end{enumerate}		
\end{Prop}
\begin{proof}
Only the last part needs a comment. Definition \eqref{eq:O-SS*} of the symmetric matrix 
$ \Omega[\bSi,\bSi^*]$ implies existence of scalar functions $f_i$ such that
\begin{equation*}
\bSi^*_i = f_i \bSi^t_i,
\end{equation*}
what inserted into the adjoint linear problem~\eqref{eq:Sij*} in comparison with~\eqref{eq:Sij} gives
\begin{equation*}
\frac{\tau f_{i(j)}}{\tau_{(ij)}} = \frac{\tau_{(j)}f_i}{\tau_{(i)}} = \frac{\sigma_{ji}f_j}{\sigma_{ij}}.
\end{equation*}
The first equality implies that, up to a constant which can be incorporated into definition of the $\tau$-function, we have $f_i= \tau \tau_{(i)}$, what leads directly to condition~\eqref{eq:S-S*} and, due to the second equality, gives constraint~\eqref{eq:C-t/s}.
\end{proof}

By applying the constraints presented in Proposition~\ref{prop:C-char-t/s} we can formulate the corresponding linear problem for the $\tau/\sigma$ form of the symmetric discrete Darboux equations and the corresponding reduction of the fundamental transformation. Both Propositions, which we present below, can be verified directly but their truth is an immediate consequence of the previous Propositions~\ref{prop:lpD-t/s},~\ref{prop:fT-t/s} and \ref{prop:C-char-t/s}. Moreover, our proof of the final Theorem~\ref{th:sDs-t/s} provides in fact such a direct verification.
\begin{Prop} \label{prop:lpsD-s/t}
	The symmetric Darboux equations in the $\tau/\sigma$ form~\eqref{eq:s-ijk}, \eqref{eq:stt} are compatibility condition of the system consisting of equations~\eqref{eq:Sij} with the constraint \eqref{eq:C-t/s}, and equations which define symmetric matrix valued potential $\Omega [\bSi]$
	\begin{equation} \label{eq:O-S}
	\Omega [\bSi]_{(i)} = \Omega[\bSi] + \tau \tau_{(i)} \bSi_i \otimes \bSi_i^t.
	\end{equation}
\end{Prop}
The corresponding form of the fundamental transformation reads as follows.
\begin{Prop} \label{prop:fTsD-t/s}
Given solution of the $\tau/\sigma$ discrete symmetric Darboux equations~\eqref{eq:s-ijk}, \eqref{eq:stt} and the corresponding solution of the linear system~\eqref{eq:Sij}, \eqref{eq:O-S}, and given additional solution $(\bT_i)$ of equations \eqref{eq:Sij}. Then the transformed solution of the $\tau/\sigma$ discrete symmetric Darboux equations and the corresponding transformed solution of the linear system read
\begin{align} \label{eq:T-tau-s}
\tilde{\tau} & = \tau \; \det \Omega[\bT],\\ \label{eq:T-sigma-s}
\tilde{\sigma}_{ij}& =\det \Omega[\bT]\left( \sigma_{ij} - \tau \tau_{(i)}\tau_{(j)}\bT_j^t \, \Omega[\bT]^{-1} \, \bT_i  \right) ,\\ \label{eq:T-Si-s}
\tilde{\bSi}_i & =\left( \bSi_i - \Omega[\bSi,\bT]  \; \Omega[\bT]^{-1} \bT_i \right) / \det \Omega[\bT]_{(i)},\\ 
\Omega[\tilde{\bSi}] & = \Omega[\bSi] - \Omega[\bSi,\bT]\, \Omega[\bT]^{-1} \Omega[\bSi,\bT]^t \label{eq:T-O-s}
\end{align}	
where the matrix valued potential $\Omega[\bSi,\bT]$ is defined by
\begin{equation} \label{eq:O-ST-s}
\Omega [\bSi,\bT]_{(i)} = \Omega[\bSi,\bT] + \tau \tau_{(i)} \bSi_i \otimes \bT_i^t.
\end{equation}
\end{Prop}
\begin{Rem}
	Again we overuse the symbol of a matrix potential $\Omega$ hoping that the proper meaning will be clear to the Reader from the context (and its arguments).
\end{Rem}
\begin{Rem}
	In the additional variable interpretation the constraints on the data $(\bT_i)$ and $\Omega[\bT]$ of the symmetric reduction of fundamental transformation simply read
	\begin{equation}
\sigma_{ia} = \sigma_{ai}, \qquad \sigma_{ab} = \sigma_{ba} .
	\end{equation}
\end{Rem}

To close the Section let us discuss transition to the discrete CKP system~\ref{taud}. By squaring equation \eqref{eq:s-ijk} we get 
\begin{equation} \label{eq:tau-s}
\tau(\tau_{(i)} \tau_{(jk)} + \tau_{(j)} \tau_{(ik)} + \tau_{(k)} \tau_{(ij)} - \tau \tau_{(ijk)}) - 2 \tau_{(i)} \tau_{(j)} \tau_{(k)} = 2 \sigma_{ij} \sigma_{jk} \sigma_{ik},
\end{equation}
which squared once again gives the discrete CKP equation~\eqref{taud}.
	Equations~\eqref{eq:tau-s} together with \eqref{eq:s-ijk} form a legitimate recurrence system in the case of symmetric reduction. In~\cite{Atkinson}, instead of equations \eqref{eq:s-ijk}, the following system
	\begin{equation} \label{eq:tau-s2}
	\tau_{(k)}(\tau_{(k)} \tau_{(ij)} -\tau_{(i)} \tau_{(jk)} - \tau_{(j)} \tau_{(ik)} - \tau \tau_{(ijk)}) + 2 \tau \tau_{(ik)} \tau_{(jk)} = 2 \sigma_{ij(k)} \sigma_{jk} \sigma_{ik},
	\end{equation}
	was considered, which can be obtained by multiplying \eqref{eq:s-ijk}  by $\sigma_{ik}\sigma_{kj}$.
	
	Following~\cite{Atkinson} let us recall the way back. The discriminant, with respect to $\tau_{(ijk)}$, of the (multi)quadratic equation~\eqref{taud} is
	\begin{equation}
	16(\tau_{(i)} \tau_{(j)} - \tau \tau_{(ij)}) (\tau_{(i)} \tau_{(k)} - \tau \tau_{(ik)})(\tau_{(j)} \tau_{(k)} - \tau \tau_{(jk)}),
	\end{equation}
	while the discriminant with respect to $\tau_{(ij)}$
	reads
\begin{equation}
16(\tau_{(i)} \tau_{(j)} - \tau \tau_{(ij)})_{(k)} (\tau_{(i)} \tau_{(k)} - \tau \tau_{(ik)})(\tau_{(j)} \tau_{(k)} - \tau \tau_{(jk)}).
\end{equation}
The discriminants motivate introduction of the functions~$\sigma_{ij}$ by~\eqref{eq:stt}. With such definitions the corresponding formulas expressing $\tau_{(ijk)}$ and $\tau_{(ij)}$ give rise to equations~\eqref{eq:tau-s} and \eqref{eq:tau-s2} respectively, where the eventual sign ambiguities are fixed by demanding multidimensional consistency of the system on four-dimensional hypercube.

\subsection{The $\tau/\sigma$ form of the discrete symmetric Darboux equations with self-consistent sources}
In the present notation Definition~\ref{def:sDs} of discrete symmetric Darboux equations with sources reads as follows (for completeness we repeat some of the equations previously written). 
\begin{Def} \label{def:sDs-t/s}
	The $\tau/\sigma$ form of the symmetric discrete Darboux equations with self-consistent sources is the system consisting of:
	\begin{enumerate}
		\item the corresponding form of the symmetric discrete Darboux equations \eqref{eq:stt} and \eqref{eq:s-ijk},
		\begin{align*} \tag{\ref{eq:stt}}
		\sigma_{ij}^2 & = \tau_{(i)} \tau_{(j)} - \tau\tau_{(ij)},\\	
		\tag{\ref{eq:s-ijk}}	
		\tau \sigma_{ij(k)} & = \tau_{(k)} \sigma_{ij} + \sigma_{ik}\sigma_{kj},
		\end{align*}
		\item definition of the source terms $(\bT_i)$ via the linear problem
		\begin{equation} \label{eq:lp-T-t/s}
		\tau_{(ij)}\bT_{i(j)}  = \tau_{(i)} \bT_i + \sigma_{ij} \bT_j,
		\end{equation}
		\item the additional part of the evolution governed by self-consistent sources
		\begin{align} \label{eq:t-K}
		\tau_{(K)} & = \tau \det \Omega[\bT],\\ \label{eq:s-K}
		\sigma_{ij(K)}& = \left( \sigma_{ij} - \tau \tau_{(i)}\tau_{(j)}\bT_j^t \, \Omega[\bT]^{-1} \, \bT_i  \right) \det \Omega[\bT] ,
		\end{align}
		with the symmetric matrix potential $\Omega[\bT] $ given by 
		\begin{equation}
		\Omega[\bT]_{(i)} = \Omega[\bT]  + \tau \tau_{(i)} \bT_i \otimes \bT^t_i \,.
		\end{equation}.
	\end{enumerate}
\end{Def}
\begin{Th} \label{th:sDs-t/s}
	The $\tau/\sigma$ symmetric discrete Darboux system with sources is the compatibility condition of the linear system consisting of:
	\begin{enumerate}
		\item the linear problem \eqref{eq:Sij} together with existence of the symmetric matrix potential $\Omega[\bSi]$ given by equations~\eqref{eq:O-S}
		\begin{align*}	\tag{\ref{eq:Sij}}
		\tau_{(ij)}\bSi_{i(j)} & = \tau_{(i)} \bSi_i + \sigma_{ij} \bSi_j ,\\
		\tag{\ref{eq:O-S}}
		\Omega [\bSi]_{(i)} & = \Omega[\bSi] + \tau \tau_{(i)} \bSi_i \otimes \bSi_i^t,
		\end{align*}
\item existence of the matrix potential $\Omega[\bSi,\bT]$ given by equations \eqref{eq:O-ST-s}
\begin{equation*} \tag{\ref{eq:O-ST-s}}
\Omega [\bSi,\bT]_{(i)} = \Omega[\bSi,\bT] + \tau \tau_{(i)} \bSi_i \otimes \bT_i^t,
\end{equation*}
\item supplementary evolution of the wave functions $(\bSi_i )$ and of the potential $\Omega[\bSi]$ governed by source terms
\begin{align} \label{eq:Si-K}
\bSi_{i(K)} & =\left( \bSi_i - \Omega[\bSi,\bT]  \; \Omega[\bT]^{-1} \bT_i \right) / \det \Omega[\bT]_{(i)},\\ \label{eq:O-K}
\Omega[\bSi]_{(K)} & = \Omega[\bSi] - \Omega[\bSi,\bT]\, \Omega[\bT]^{-1} \Omega[\bSi,\bT]^t .
\end{align}
\end{enumerate}
\end{Th}
\begin{proof}
	Each of the parts of the linear problem will give the corresponding part of the $\tau/\sigma$ symmetric discrete Darboux equations with sources. Notice moreover, that the first part of the proof gives direct verification of Proposition~\ref{prop:lpsD-s/t}, while the third part gives direct verification of Proposition~\ref{prop:fTsD-t/s}; compare equations~\eqref{eq:t-K}-\eqref{eq:s-K} and~\eqref{eq:Si-K}-\eqref{eq:O-K} with transformations formulas~\eqref{eq:T-tau-s}-\eqref{eq:T-O-s}.
	
	(1) Compatibility of the system~\eqref{eq:O-S} in view of equations~\eqref{eq:Sij} gives equations~\eqref{eq:stt}. Then compatibility of the system~\eqref{eq:Sij} leads to
	\begin{equation*}
	(\tau_{(j)} \tau_{(k)} - \sigma_{jk}^2) \sigma_{ij(k)} = \tau_{(jk)} \tau_{(k)} \sigma_{ij} + \tau_{(jk)} \sigma_{ik} \sigma_{kj},
	\end{equation*} 
	which, by just shown equations~\eqref{eq:stt}, gives~\eqref{eq:s-ijk}.  
	
	(2) Compatibility of the system~\eqref{eq:O-ST-s} in view of equations~\eqref{eq:Sij} leads to
	\begin{equation*}
	(\tau_{(i)} \tau_{(j)} - \sigma_{ij}^2) \bT_{i(j)} = \tau \left( \tau_{(i)} \bT_{i} + \sigma_{ij}\bT_j \right),
	\end{equation*} 
	which, by equations~\eqref{eq:stt}, gives the linear problem~\eqref{eq:lp-T-t/s}. 
	
	(3) This implies existence of the symmetric matrix potential $\Omega[\bT] $. In consequence, by the standard technique of bordered matrices~\cite{Hirota-book}, we have the following evolution of its determinant
	\begin{equation}
	\det \Omega[\bT]_{(i)} = \det \Omega[\bT]  \left( 1 + \tau \tau_{(i)} \bT^t_i \,  \Omega[\bT]^{-1} \bT_i \right).
	\end{equation}
	Then, equation~\eqref{eq:Si-K} and its equivalent version
	\begin{equation}
	\bSi_{i(K)}  =\left( \bSi_i - \Omega[\bSi,\bT]_{(i)}  \; \Omega[\bT]^{-1}_{(i)} \bT_i \right) / \det \Omega[\bT],
	\end{equation}
	imply, together with definitions of the potentials $\Omega[\bSi,\bT]$ and $\Omega[\bT]$, that
	\begin{gather*}
(\tau \det \Omega[\bT]) \bSi_{i(K)} \otimes (\tau \det \Omega[\bT])_{(i)} \bSi_{i(K)}^t =\\=
\left(\Omega[\bSi] - \Omega[\bSi,\bT]\, \Omega[\bT]^{-1} \Omega[\bSi,\bT]^t \right)_{(i)} - \left(\Omega[\bSi] - \Omega[\bSi,\bT]\, \Omega[\bT]^{-1} \Omega[\bSi,\bT]^t \right) = 
\Omega[\bSi]_{(Ki)} - \Omega[\bSi]_{(K)}.
	\end{gather*}
Its compatibility with equation~\eqref{eq:O-S} shifted in the source direction
\begin{equation*}
\Omega [\bSi]_{(iK)} - \Omega[\bSi]_{(K)} = \tau_{(K)} \tau_{(iK)} \bSi_{i(K)} \otimes \bSi_{i(K)}^t,
\end{equation*}
gives the corresponding evolution~\eqref{eq:t-K} of the $\tau$-function. Finally, compatibility of linear equations~\eqref{eq:Sij} and~\eqref{eq:Si-K} gives, after a few line calculation, the source direction evolution equation~\eqref{eq:s-K}.
\end{proof}
\begin{Rem}
	It would be instructive to reformulate the $\tau/\sigma$ form of discrete symmetric Darboux equations with sources to the $\tau$ form, analogous to the discrete CKP equation~\eqref{taud} version of the symmetric discrete Darboux equations. However, one should expect that it cannot be written as a legitimate recurrence, and the transition back to the  $\tau/\sigma$ form will be necessary, like in the case of equations without sources~\cite{Atkinson}. 
\end{Rem}

 \section{Darboux equations with self-consistent sources and the symmetric reduction}
 \label{sec:cont}
In analogy to the discrete case, we can also propose the version with self-consistent sources of the standard Darboux equations and of their symmetric reduction (first considered by Darboux himself~\cite{Darboux-OS}). We will keep the notation of the previous Sections for solutions of the linear problems and matrix-valued potentials. We denote rotation coefficients using standard symbols $\beta_{ij}$ as in classical books od Darboux or Bianchi~\cite{Bianchi}. By $\partial_i$ we denote partial derivative with respect to variable $x_i$, $i=1,\dots , N$, and $\partial_K$ is the derivative with respect to the source variable. As a rule in the "natural continuous limit" procedure we change differences into derivatives neglecting other shifts. Instead of taking the limit we decided to write down the equations and study them directly.

\begin{Th} \label{th:cDs-s}
The Darboux equations with sources
\begin{align} \label{eq:cD}
\partial_k \beta_{ij} & = \beta_{ik} \beta_{kj}, \qquad k\neq i,j,\\
\label{eq:lc-Y*}
\partial_i \bY_j^* & = \beta_{ij} \bY^*_i, \\
\label{eq:lc-pi}
\partial_j \bY_i & = \beta_{ij} \bY_j , \qquad j\neq i, \\
\label{eq:cDs}
\partial_K \beta_{ij} & = -\bY^*_j \bY_i, 
\end{align}
are compatibility condition of the linear problem
\begin{align} 
\label{eq:lcD}
\partial_j \bX_i & = \beta_{ij} \bX_j,\\
\label{eq:cO-XY*}
\partial_i \Omega[\bX,\bY^*] & = \bX_i \otimes \bY_i^*, \\
\label{eq:lcs-X}
\partial_K  \bX_i & = -\Omega[\bX,\bY^*] \bY_i .
\end{align}
\end{Th}
\begin{proof}
	The first part~\eqref{eq:lcD} of the linear problem gives the standard Darboux equations~\eqref{eq:cD}. Cross-differentiation of equation~\eqref{eq:cO-XY*} implies the linear system~\eqref{eq:lc-Y*} for the sources terms $\bY_i^*$. Compatibility of equations~\eqref{eq:lcs-X} with~\eqref{eq:lcD} gives the linear system~\eqref{eq:lc-pi} for the sources $\bpi_i$ (terms at $\Omega[\bX,\bY^*]$) and the evolution~\eqref{eq:cDs} of the rotation coefficients in the source direction (terms at $\bX_j$).
\end{proof}
\begin{Cor}
	Darboux equations with sources~\eqref{eq:cD}-\eqref{eq:cDs} can be equivalently obtained as the compatibility condition of the adjoint linear problem
	\begin{align} \label{eq:lcD*}
	\partial_i \bX_j^* & = \beta_{ij}X_i^*,\\
	\partial_i \Omega[\bY,\bX^*] & = \bY_i \otimes \bX_i^*, \\
	\label{eq:lcs-X*}
	\partial_K  \bX_i^* & = -\bY_i \Omega[\bY,\bX^*]  .
	\end{align}
\end{Cor}

\begin{Rem}
	Let us modify the source direction by translation in ordinary variable $x_\ell$, and define the corresponding flow $\partial_L = \partial_\ell + \partial_K$. Then the source direction evolution equation \eqref{eq:cDs} gets modified to
\begin{equation}
\partial_L \beta_{ij} = \beta_{i\ell} \beta_{\ell j} - \bY^*_j \bY_i,
\end{equation}
which is the form of Darboux equation with sources obtained in~\cite{HuHuTam}. 
Accordingly, the corresponding parts \eqref{eq:lcs-X} or \eqref{eq:lcs-X*} of the linear problems should be changed to
\begin{align}
\partial_L  \bX_i & = \beta_{i\ell} \bX_\ell -\Omega[\bX,\bY^*] \bY_i ,\\
\partial_L  \bX_i^* & = \beta_{\ell i} \bX_\ell^* -\bY_i^* \Omega[\bY,\bX^*]  .
\end{align}
\end{Rem}

Now one can consider the symmetric reduction of the Darboux equations with sources. Recall that the symmetric reduction 
\begin{equation}
\label{eq:cD-sym-c}
\beta_{ij}=\beta_{ji}
\end{equation} 
of the Darboux equations was considered in~\cite{Darboux-OS}. In that case the linear problem~\eqref{eq:lcD} coincides with the adjoint linear problem~\eqref{eq:lcD*}. Therefore the constraint
\begin{equation}
\label{eq:cD-sym-cY}
\bY^*_i = \bY^t_i,
\end{equation}
imposed on the sources gives, together with~\eqref{eq:cD-sym-c}, the good candidate for the symmetric Darboux equations with sources. Integrability of the reduction follows from the limiting version of Proposition~\ref{prop:symetr-prop}.
\begin{Prop}
	The above constraints imposed on Darboux equations with sources are preserved during the evolution.
\end{Prop}
\begin{proof}
	Indeed, simple calculation shows that the functions
	\begin{equation}
	C_{ij} = \beta_{ij} - \beta_{ji}, \qquad C_i = \bY^*_i - \bY^t_i,
	\end{equation} 
	evolve according to
	\begin{align*}
	\partial_k C_{ij} & = C_{ik} \beta_{kj} - C_{jk} \beta_{ki},\\
	\partial_K C_{ij} & = C_i \bY^*_j - C_j \bY^*_i, \\
	\partial_j C_i & = C_j \beta_{ji} - C_{ij} \bY^t_j.
	\end{align*}
\end{proof}
Finally, let us write down the equations together with the corresponding linear problem. The following result can be easily checked along lines of the proof of Theorem~\ref{th:cDs-s}.
\begin{Th}
The symmetric Darboux equations with sources (here we assume that the condition~\eqref{eq:cD-sym-c} holds)
\begin{align} \label{eq:cD-sym}
\partial_k \beta_{ij} & = \beta_{ik} \beta_{kj}, \qquad \\
\label{eq:lc-Y-sym}
\partial_j\bY_i & = \beta_{ij} \bY_j, \\
\label{eq:cDs-sym}
\partial_K \beta_{ij} & = -\bY^t_j \bY_i, 
\end{align}
are compatibility condition of the linear system
\begin{align} 
\tag{\eqref{eq:lcD}}
\partial_j \bX_i & = \beta_{ij} \bX_j,\\
\label{eq:cO-XY-sym}
\partial_i \Omega[\bX,\bY] & = \bX_i \otimes \bY_i^t, \\
\label{eq:lcs-X-sym}
\partial_K  \bX_i & = -\Omega[\bX,\bY] \bY_i.
\end{align}
\end{Th}

\section{Conclusion and open problems}
We presented construction of the non-commutative discrete Darboux system with  self-consistent sources together with its linear problem and its adjoint. We used the methods proposed in \cite{doliwakp} applied there to Hirota's discrete KP equation. Our main aim however, was to construct the C-symmetric reduction of the discrete Darboux equations with sources. It can be done by application of the methods of~\cite{doliwakp} to the discrete version of the symmetric constraint~\cite{DS-sym} and the corresponding reduction of the fundamental transformation~\cite{MM}. To simplify the resulting equations we reformulated the (commutative) discrete Darboux equations and the fundamental transformation to the $\tau/\sigma$ form. As a by-product we gave also symmetric Darboux equations with self-consistent sources together with the corresponding linear problem. 

This paper presents another example to claim that the method proposed in \cite{doliwakp}, by means of binary Darboux transformations and the technique of introducing additional variables, is a general one to construct discrete integrable systems with sources. It would be instructive to apply the method to other integrable systems, in particular to Miwa's discrete BKP equation~\cite{Miwa} --- the corresponding reductions of the discrete Darboux equations and of the fundamental transformation are given in~\cite{BQL}. We stress that the simplicity of our method is visible on the level of discrete systems~\cite{DSI}, especially if we use their incidence-geometric interpretation~\cite{MQL,CDS,TQL,CQL,BQL,Dol-Des}. 

Let us discuss finally another generalization of our approach. In our method we used only one collective variable (source variable) obtained by simultaneous shift in $K$ additional "small" variables, leaving other $N$~discrete variables untouched. This corresponds to partition $(1^N,K)$ of $N+K$. However other partitions are possible, and one can imagine integrable discrete equations obtained from others by collective shifts in several groups of "small" variables. In fact, the transition from Hirota's discrete KP system to the discrete Darboux equations can be done in such a way~\cite{Dol-Des,Dol-Tampa}.

\section*{Acknowledgments}
A.D. was supported by National Science Centre, Poland, under grant 2015/19/B/ST2/03575 \emph{Discrete integrable systems -- theory and applications}.
R.L. was supported by the National Natural Science Foundation of China (11471182).


\bibliographystyle{amsplain}

\end{document}